\documentclass[12pt, draftclsnofoot, onecolumn,romanappendices]{IEEEtran}
\usepackage{bm}
\usepackage{amsmath}
\usepackage{amsthm}
\usepackage{amsbsy}
\usepackage{epsfig}
\usepackage{latexsym}
\usepackage{psfrag}
\usepackage{amssymb} 
\usepackage{subfig} 
\usepackage{placeins}
\usepackage{hyperref}

\newtheorem{theorem}{Theorem}

\newtheorem{proposition}{Proposition}

\DeclareMathAlphabet{\mathbit}{OML}{cmr}{bx}{it}
\DeclareMathAlphabet{\mathsf}{OT1}{cmss}{m}{n}
\DeclareMathAlphabet{\mathTXf}{OT1}{cmss}{bx}{it}

\DeclareMathOperator*{\argmin}{argmin}

\DeclareMathOperator{\DoF}{DoF}

\DeclareMathOperator{\vect}{vect}

\DeclareMathOperator{\DCSI}{DCSI}
 
\DeclareMathOperator{\CN}{\mathcal{N}_{\mathbb{C}}}

\newcommand{\bA}{\mathbf{A}}

\newcommand{\bG}{\mathbf{G}} 
\newcommand{\bH}{\mathbf{H}}

\newcommand{\bN}{\mathbf{N}}

\newcommand{\bR}{\mathbf{R}} 
 
\newcommand{\bT}{\mathbf{T}} 
\newcommand{\bU}{\mathbf{U}} 
\newcommand{\bV}{\mathbf{V}} 
\newcommand{\bW}{\mathbf{W}} 
 
\newcommand{\bY}{\mathbf{Y}}

\newcommand{\LB}{\left(}
\newcommand{\RB}{\right)}
\newcommand{\LSB}{\left[}
\newcommand{\RSB}{\right]}

\newcommand*{\dotleq}{\mathrel{\dot{\leq}}}
\newcommand*{\dotgeq}{\mathrel{\dot{\geq}}}

\newcommand{\e}{{\operatorname{e}}}
\newcommand{\I}{\mathbf{I}} 

\newcommand{\Fro}{{\text{F}}}

\newcommand{\E}{{\mathrm{E}}}
\newcommand{\tot}{{\mathrm{tot}}}

\newcommand{\He}{{{\mathrm{H}}}}

\newcommand{\EVD}{{\mathrm{EVD}}}  

\theoremstyle{remark}
\newtheorem{remark}{Remark} 

\begin{document} 

\title{Degrees of Freedom of Certain Interference Alignment Schemes with Distributed CSIT}
\author{\IEEEauthorblockN{Paul de Kerret\IEEEauthorrefmark{1}, Jakob Hoydis\IEEEauthorrefmark{4}, and David Gesbert\IEEEauthorrefmark{1}}

\IEEEauthorblockA{\IEEEauthorrefmark{1}
Eurecom, Campus SophiaTech, 450 Route des Chappes, 06410 Biot, France {\tt{\{dekerret,gesbert\}@eurecom.fr}}}

\IEEEauthorblockA{\IEEEauthorrefmark{4}
Bell Laboratories, Alcatel-Lucent, Lorenzstr. 10, 70435
Stuttgart, Germany, {\tt{jakob.hoydis@alcatel-lucent.com}}}}

%
\author{\IEEEauthorblockN{Paul de Kerret\IEEEauthorrefmark{1}, Maxime Guillaud\IEEEauthorrefmark{3}, and David Gesbert\IEEEauthorrefmark{1}}

\IEEEauthorblockA{\IEEEauthorrefmark{1}
Eurecom, Campus SophiaTech, 450 Route des Chappes, 06410 Biot, France\\ {\normalsize \tt{\{\href{mailto:dekerret@eurecom.fr}{dekerret},\href{mailto:gesbert@eurecom.fr}{gesbert}\}@eurecom.fr}}}

\IEEEauthorblockA{\IEEEauthorrefmark{3}
Institute of Telecommunications - Vienna University of Technology\\ Gu{\ss}hausstra{\ss}e 25 / E389, A-1040 Vienna, Austria\\ \href{mailto:guillaud@tuwien.ac.at}{\normalsize \tt{guillaud@tuwien.ac.at}}}}

\maketitle

\begin{abstract}
In this work, we consider the use of interference alignment (IA) in a MIMO interference channel (IC) under the assumption that each transmitter (TX) has access to channel state information (CSI) that generally differs from that available to other TXs. This setting is referred to as \emph{distributed CSIT}. In a setting where CSI accuracy is controlled by a set of power exponents, we show that in the static $3$-user MIMO square IC, the number of degrees-of-freedom (DoF) that can be achieved with distributed CSIT is at least equal to the DoF achieved with the worst accuracy taken across the TXs \emph{and} across the interfering links. We conjecture further that this represents exactly the DoF achieved. This result is in strong contrast with the \emph{centralized CSIT} configuration usually studied (where all the TXs share the same, possibly imperfect, channel estimate) for which it was shown that the DoF achieved at receiver (RX)~$i$ is solely limited by the quality of its \emph{own} feedback. This shows the critical impact of CSI discrepancies between the TXs, and highlights the price paid by distributed precoding. 
\end{abstract}
\IEEEpeerreviewmaketitle

\let\thefootnote\relax\footnotetext{We acknowledge the support of the Newcom\# Network of Excellence in Wireless Communications, under the 7th Framework Program of the European Commission (EC), as well as of the Franco-Austrian EGIDE-\"OAD ``Amadeus'' Programme, under grant \#FR05/2012. M. Guillaud was also supported by the FP7 HIATUS project of the EC and by the Austrian Science Fund (FWF) through grant NFN SISE (S106). David Gesbert and Paul de Kerret acknowledge support from the Celtic European project SHARING.

Part of this work will be presented at the IEEE conference SPAWC, Darmstadt, June 2013.} 

\section{Introduction} 

It has recently been shown that an improvement in the DoF achieved over certain multi-user channels could be obtained by designing the transmission scheme such that interference aligns at the RXs \cite{MaddahAli2008,Cadambe2008,Gou2010}. The first IA scheme was based on the coding of the user's data symbols across multiple orthogonal dimensions (called \emph{symbol extension}) to align the interference over half the dimensions, thus leaving half the dimensions free of interference \cite{Cadambe2008}. IA has then be applied to MIMO ICs without symbol extension and has become the center of a strong interest. A large number of iterative IA algorithms have then been provided for that setting (see \cite{Gomadam2008,Schmidt2009,Peters2011}, among others).

One of the main obstacles for the practical use of IA comes from the need to gather the CSI relative to the global multi-user channel. Indeed, the resources available for feedback are very limited and make the obtaining of the multi-user CSI at the TX (CSIT) in a timely manner especially challenging~\cite{Sesia2011}. 

Consequently, the study of how CSIT requirements for IA methods can somehow be alleviated has become an active research topic in its own right~\cite{Gomadam2008,Thukral2009,Krishnamachari2010,Schmidt2009,Suh2011,ElAyach2012}. Another line of work consists in studying the minimal number of CSI quantization bits that should be conveyed to the TXs to achieve some given DoF using IA \cite{Thukral2009,Krishnamachari2010,Rezaee2013a}. It should be noted that in all these works, every one of the TXs is assumed to be provided with the \emph{same} quantized CSIT, meaning that the imperfect estimates are perfectly shared between the TXs, which we call the \emph{centralized} CSIT configuration, since this setting is equivalently obtained when all the precoders are computed centrally and then shared to the TXs.

Since the interfering TXs in an IC are usually not colocated, this assumption is likely to be breached. Indeed, each TX is likely to receive its channel estimate via a different feedback channel. For example, if the CSIT is obtained via an analog feedback broadcast from the RXs, as in \cite{ElAyach2012}, each TX will receive a different estimation of the multi-user channel with a priori different accuracies. An alternative possibility, currently envisioned for future LTE systems, consists in letting each RX feedbacks its CSI to its serving TX which then forwards it to the other TXs \cite{Sesia2011,Rezaee2013b}. In that setting as well, the sharing step leads in most cases to a CSIT aging, or requires further quantization. In both scenarios, each TX receives its \emph{own} estimate of the multi-user channel based on which it computes its precoder without additional communications with the other TXs. This case has been first denoted in \cite{Zakhour2010a,dekerret2011_ISIT_journal} as the \emph{distributed CSIT} configuration. 

The distributed CSIT scenario has recently gained in interest with the developement of TX cooperation in wireless networks. In \cite{Aggarwal2011}, the IC is studied when each TX has a \emph{local view} of the IC. Specifically, each TX has a perfect knowledge of the channel coefficients inside a given neighborhood and no knowledge of the other coefficients. In \cite{Aggarwal2012}, the same concept of local view is discussed this time with rounds of message passing. Extending the model of channel with state from \cite{Gelfand1980,Costa1983}, the transmission in multiple-access channels (MAC) has also been studied when each TX has access to a different CSIT \cite{Lapidoth2010,Como2011}. Going back to IA, the scenario where the TXs have only an \emph{incomplete} knowledge of the multi-user channel in the sense that the TXs do not have the knowledge of all the channel coefficients, is studied in \cite{dekerret2012_ISWCS_journal,dekerret2013_WCM}. It is shown that the IA algorithm
 s can be modified to achieve IA in some cases using only incomplete CSIT.

Nevertheless, the feedback/quantization requirements for IA with distributed CSIT have never been studied. Thus, we investigate here how the works \cite{Thukral2009,Krishnamachari2010} deadling with IA in the centralized CSIT configuration extend to the distributed CSIT case. Specifically, our main contributions are as follows:
\begin{itemize}
\item In a general MIMO IC, we provide a sufficient criterion on the accuracy of the precoder design to achieve the maximum DoF.
\item Studying the particular $3$-User MIMO square setting, we provide a closed-form expression for the achieveable DoF. It is shown to depend on the worst accuracy across the TXs and the channel elements.  
\end{itemize}

\paragraph*{Notations} We write~$x\doteq y$ to represent the exponential equality in the SNR~$P$, i.e., $\lim_{P\rightarrow \infty} \log_2(x)/\log_2(P)=\lim_{P\rightarrow \infty}\log_2(y)/\log_2(P)$. The inequalities $\dotleq$ and $\dotgeq$ are defined similarly. $\mathcal{N}(0,1)$ is used to represent the complex circularly symmetric zero-mean unit-variance Gaussian distribution. $\lambda_i(\mathbf{A})$ denotes the $i$th eigenvalue (orderered by decreasing absolute value) of the diagonalizable matrix~$\mathbf{A}$ while $\lambda_{\min}(\mathbf{A})$ denotes the eigenvalue with the smallest absolute value. $\EVD(\bA)$ denotes the eigenbasis of the diagonalizable matrix~$\bA$. $\vect(\bA)$ is the vector made of the stacked columns of the matrix~$\bA$. $\E_{\mathcal{A}}[\cdot]$ denotes the expectation over the subspace~$\mathcal{A}$ and $\Pr(\mathcal{A})$ the probability of the subspace~$\mathcal{A}$.

\section{System Model}\label{se:SM}

\subsection{MIMO interference channel}\label{se:SM:DCSI}

We consider a conventional static MIMO IC with $K$~users \cite{Peters2011} and assume that each TX has its \emph{own} CSI in the form of an imperfect estimate of the whole multi-user channel state. TX~$j$ is equipped with $M_j$~antennas and RX~$i$ with $N_i$~antennas. The antenna configuration is supposed to be tightly-feasible in the sense that the number of antennas available is the minimal one which allows to achieve the DoF desired at every user \cite{dekerret2012_ISWCS_journal}. The channel from TX~$j$ to RX~$i$ is represented by the channel matrix~$\mathbf{H}_{i,j}\in \mathbb{C}^{N_i\times M_j}$ with its elements distributed according to a continuous distribution which ensures that all the sub-matrices are almost surely full rank. We denote by~$\mathcal{H}$ the space of all possible channel realizations. 
Since interference alignment is invariant by scaling (by a non-zero complex scalar) of the channel matrices, we further define $\tilde{\bH}_{i,j}\triangleq \e^{\jmath\phi_{i,k}} \frac{\mathbf{H}_{i,j}}{\|\mathbf{H}_{i,j}\|_{\Fro}}$,
where~$\phi_{i,k}\in\mathbb{R}$ is chosen so as to let the first element of~$\vect(\tilde{\bH}_{i,k})$ be real valued.

The global multi-user channel matrix is denoted by $\mathbf{H}\in \mathbb{C}^{N_{\tot}\times M_{\tot}}$ with $N_{\tot}\triangleq \sum_{i=1}^ KN_i$ and $M_{\tot}\triangleq \sum_{i=1}^ KN_i$, and defined as
\begin{equation}
\mathbf{H}\triangleq\begin{bmatrix}
\mathbf{H}_{1,1}&\mathbf{H}_{1,2}&\ldots&\mathbf{H}_{1,K}\\
\mathbf{H}_{2,1}&\mathbf{H}_{2,2}&\ldots&\mathbf{H}_{2,K}\\
\vdots&\vdots&\ddots&\vdots\\
\mathbf{H}_{K,1}&\mathbf{H}_{K,2}&\ldots&\mathbf{H}_{K,K}
\end{bmatrix}.
\label{eq:SM_1}
\end{equation}
The matrix~$\tilde{\bH}$ is defined similarly from the matrices~$\tilde{\bH}_{i,k}$. 

Assume that TX~$j$ uses the precoder~$\mathbf{T}_j\triangleq \sqrt{P} \bU_j\in\mathbb{C}^{M_j\times d_j}$ with $\|\bU_j\|^2_{\Fro}=1$ to transmit the data symbol $\bm{s}_j\in \mathbb{C}^{d_j}$ (i.i.d. $\mathcal{N}(0,1)$) to RX~$j$. Hence, the precoder fulfills the per-TX power~$\|\bT_j\|_{\Fro}^2=P$.

The received signal $\bm{y}_i\in \mathbb{C}^{N_i}$ at RX~$i$ is
\begin{equation}
\bm{y}_i=\sqrt{P}\mathbf{H}_{i,i}\bU_i \bm{s}_i+\sqrt{P}\sum_{j=1,j\neq i}^K \mathbf{H}_{i,j}\bU_j  \bm{s}_j+\bm{\eta}_i
\label{eq:SM_2}
\end{equation}
where $\bm{\eta}_i\in \mathbb{C}^{N_i}$ is the noise at RX~$i$ and has its elements i.i.d. $\CN(0,1)$. The received signal $\bm{y}_i$ is then processed by a RX filter $\mathbf{G}_i^{\He}\in \mathbb{C}^{d_i\times N_i}$ with $\|\bG_i\|^2_{\Fro}=1$.

The average rate achieved at user~$i$ is written as
\begin{equation}
R_i=\E_{\mathcal{H},\mathcal{W}}\left[\log_2\left|\I_{d_i}+P\bar{\bR}_i^{-1}\bG_i^{\He}\bH_{i,i}\bU_i\bU_i^{\He}\bH_{i,i}^{\He}\bG_i\right|\right]
\label{eq:SM_3}
\end{equation}
where
\begin{equation}
\bar{\bR}_i= \I_{d_i}+ P \sum_{\ell=1,\ell\neq i}^K\bG_i^{\He}\bH_{i,\ell}\bU_{\ell}\bU_{\ell}^{\He}\bH_{i,\ell}^{\He}\bG_i
\label{eq:SM_4}
\end{equation}
and $\E_{\mathcal{H},\mathcal{W}}[\cdot]$ denotes the expectation over the channel matrices and the channel estimation errors according to the feedback model described in Subsection~\ref{se:SM:DCSI}.
The DoF at user~$i$, or prelog factor, is then defined as
\begin{equation}
\DoF_i=\lim_{P\rightarrow \infty} \frac{R_i}{\log_2(P)}.
\label{eq:SM_5}
\end{equation}

\subsection{Distributed CSIT and distributed precoding}\label{se:SM:DCSI}
Let us assume that TX~$j$ receives its own estimate of the channel from TX~$k$ to RX~$i$. We denote  this estimate by~$\tilde{\bH}_{i,k}^{(j)}$, assumed to have the same properties (unit norm and real-valued first coefficient) as $\tilde{\mathbf{H}}_{i,k}$. Furthermore, similar to \eqref{eq:SM_1}, we let $\tilde{\bH}^{(j)}$ denote the channel state information available at TX~$j$.
In the sequel, we assume that each TX independently computes its own solution of the IA problem based on its own CSI. 
Specifically, TX~$j$ computes the solution (in terms of the precoders and receive filters $\bU_k^{(j)}, k=1\ldots K$ and $\bG_i^{(j)}, i=1\ldots K$) of its own IA problem based on $\tilde{\bH}^{(j)}$,
\begin{equation}
 (\bG_i^{(j)})^{\He} \tilde{\bH}_{i,k}^{(j)} \bU_k^{(j)} = {\bf 0}_{d_i\times d_j} \quad \forall k\neq i
\end{equation}
where~$\bU_k^{(j)}$ is the precoder designed to be used by TX~$k$ and $\bG_i^{(j)}$ is the receive filter assumed at RX~$i$.
However, since the TXs are not colocated and do not exchange further informations, $\bU_j^{(j)}$ is used for the actual transmission at TX~$j$, while the $\bU_i^{(j)}, i\neq j$ are discarded, such that considering the whole network we have
\begin{equation} 
\bU_j=\bU_j^{(j)},\qquad \forall j.
\end{equation}

This distributed CSIT setting is depicted and compared to the centralized CSIT configuration in Fig.~\ref{distributed-centralized}.

\subsection{Imperfect CSI model}

\begin{figure}[ht!]
\centering
\subfloat[][]{
\label{centralized}
\includegraphics[width=0.4\columnwidth]{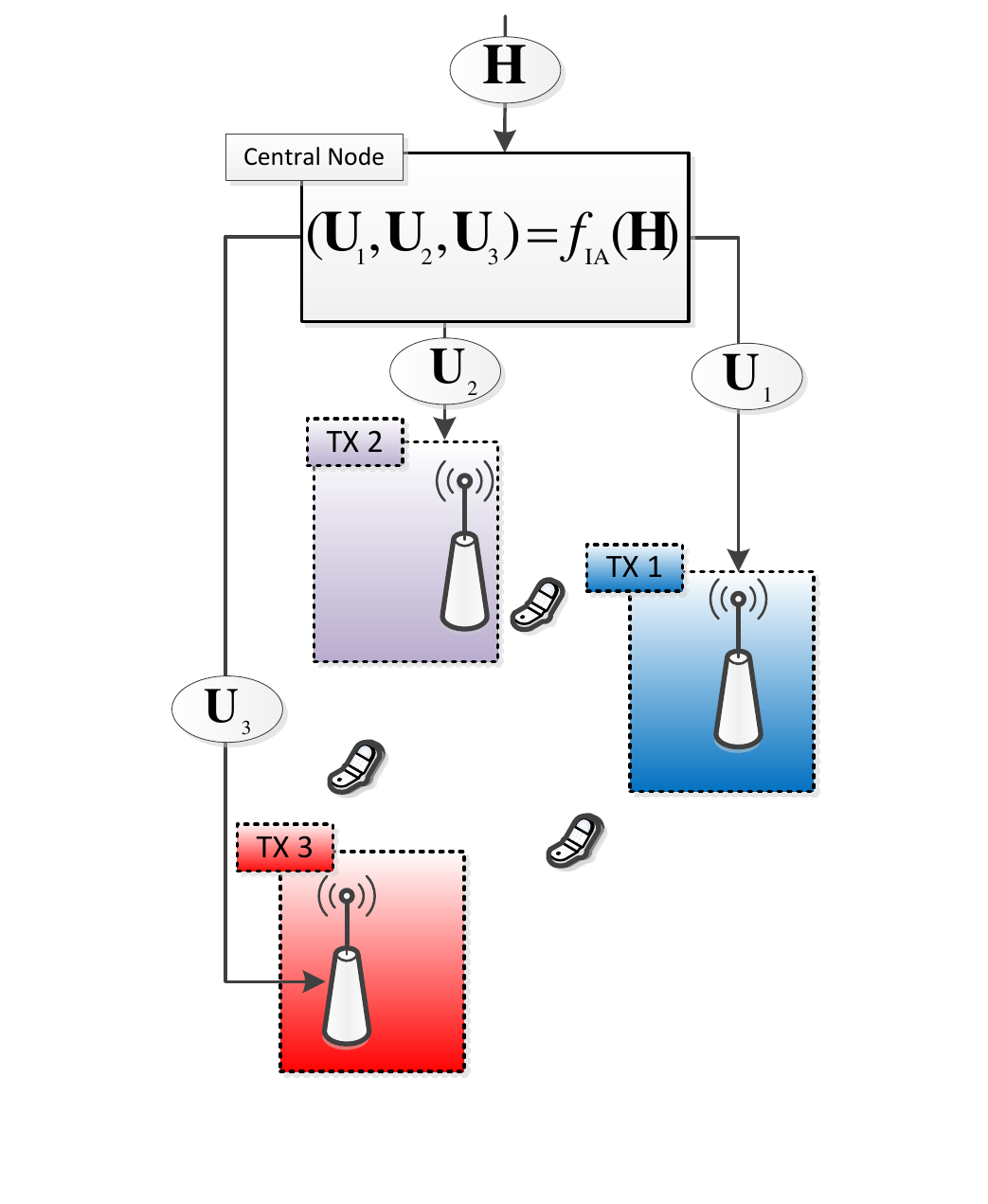}}
\hspace{8pt}  
\subfloat[][]{\label{distributed}
\includegraphics[width=0.5\columnwidth]{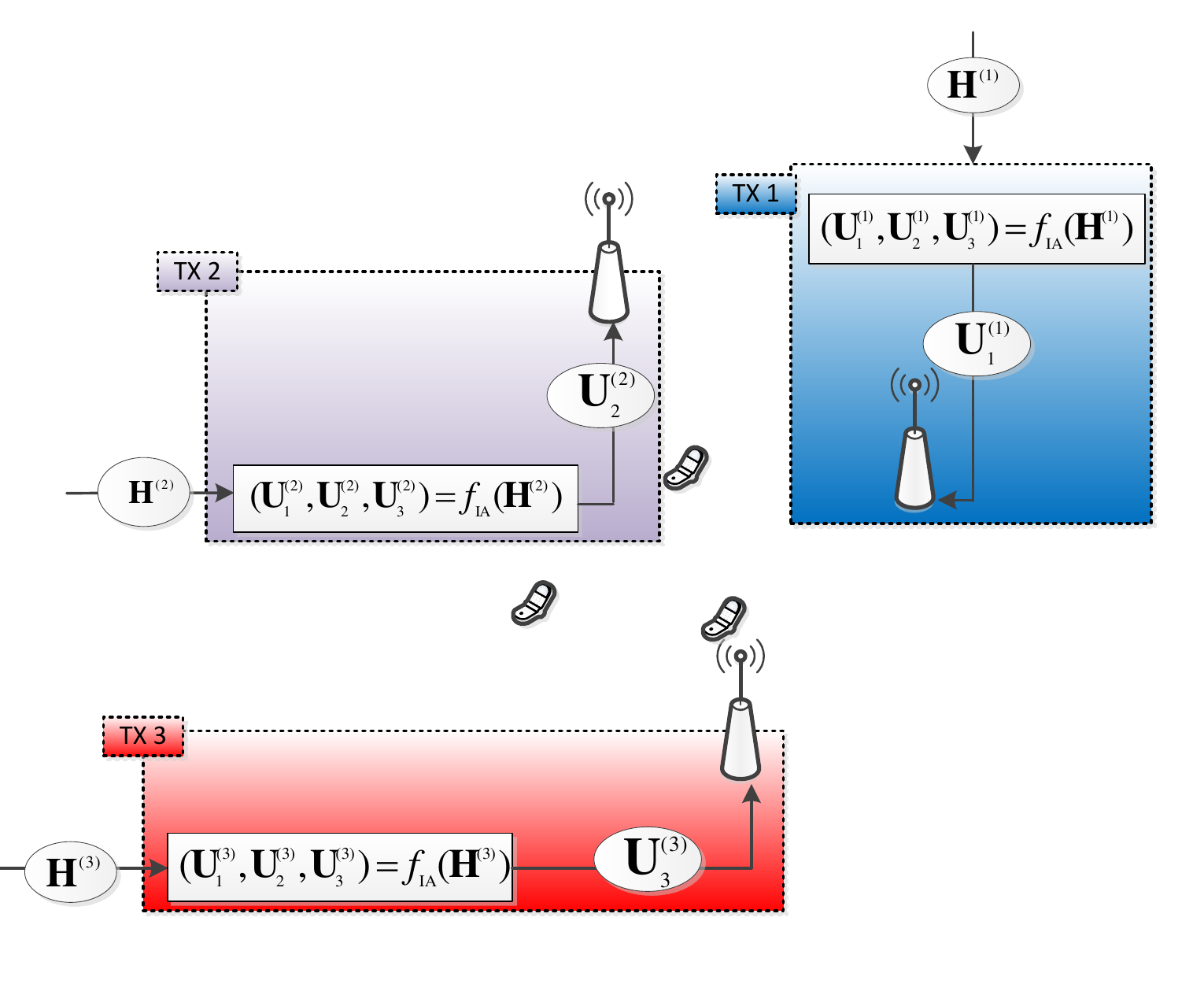}}
\caption[]{IA precoding with centralized precoding/CSIT is symbolically represented in Figure~\subref{centralized} while Figure~\subref{distributed} represents IA with distributed precoding/CSIT.} 
\label{distributed-centralized}%
\end{figure}

Let us assume that $\tilde{\bH}_{i,k}^{(j)}$ results from the quantization of $\tilde{\mathbf{H}}_{i,k}$, using a quantization scheme using $B_{i,k}^{(j)}$ bits according to
\begin{equation}
\tilde{\bH}_{i,k}^{(j)}=\argmin_{\vect(\bW)\in\mathcal{W}_{i,k}^{(j)}}\left\|\tilde{\mathbf{H}}_{i,k} -\bW\right\|_{\Fro},\qquad \forall k,i,j,
\label{eq:SM_5}
\end{equation}
where $\mathcal{W}_{i,k}^{(j)}$ contains~$2^{B_{i,k}^{(j)}}$ vectors of size~$\mathbb{C}^{N_i M_k}$ isotropically distributed over the unit-sphere and rotated to have their first element real-valued.
We further define 
\begin{align}
(\sigma_{i,k}^{(j)})^2&\triangleq \E_{\mathcal{H},\mathcal{W}}\left[\left\|\tilde{\bH}_{i,k}^{(j)}-\tilde{\bH}_{i,k}\right\|_{\Fro}^2\right] \quad \mathrm{and}\\
\bN_{i,k}^{(j)}&\triangleq  \frac{\tilde{\bH}_{i,k}^{(j)}-\tilde{\bH}_{i,k}}{\sigma_{i,k}^{(j)}},
\label{eq:SM_6}
\end{align}
where~$\E_{\mathcal{W}}[\cdot]$ denotes the expectation over the random codebooks. It then gives
\begin{equation}
\tilde{\bH}_{i,k}^{(j)}=\tilde{\bH}_{i,k}+\sigma_{i,k}^{(j)}\bN_{i,k}^{(j)}.
\label{eq:SM_7}
\end{equation}
Since there is no confusion possible we use the short notation $\E_{\mathcal{H}}[\cdot]$ instead of $\E_{\mathcal{H},\mathcal{W}}[\cdot]$.

Due to the adopted normalization, the quantization scheme \eqref{eq:SM_5} corresponds to the Grassmannian quantization over the Grassmannian space, similar to that used in \cite{Jindal2006,Krishnamachari2010}.
Using this property and the results from \cite{Dai2008,dekerret2011_ISIT_journal}, the variance of the estimation error can be related to the number of quantization bits as follows.
\begin{proposition}[{\cite[Theorem~$2$]{Dai2008}}]
When the size $L_{i,k}^{(j)}=2^{B_{i,k}^{(j)}}$ of the random codebook is sufficiently large, it then holds that
\begin{equation} 
(\sigma_{i,k}^{(j)})^2= C_{i,k}^{(j)}2^{-B_{i,k}^{(j)}/(N_i M_k-1)}
\label{eq:SM_8}
\end{equation} 
for some constant~$C_{i,k}^{(j)}>0$.
\label{App_distorsion}
\end{proposition}

 
With centralized CSIT, is a well known result \cite{Caire2010,Thukral2009} that the number of quantization bits should scale with the SNR in order to achieve a positive DoF. Hence, we define the \emph{CSIT scaling coefficients} $A_{i,k}^{(j)}$ as 
\begin{equation}
 \quad A_{i,k}^{(j)}\triangleq \lim_{P\rightarrow \infty}\frac{B_{i,k}^{(j)}}{B_{i,k}^\star},\forall k,i,
\label{eq:ST_3}
\end{equation}
where we have defined
\begin{equation}
B_{i,k}^ \star\triangleq (N_i M_k-1)\log_2(P).
\end{equation}
The pre-log coefficient~$N_i M_k-1$ corresponds to the number of channel coefficients to feedback ater normalization of the channel matrix. $B_{i,k}^\star$ is the number of bits which corresponds to a quantization error decreasing as $P^{-1}$, which is essentially perfect in terms of DoF \cite{Jindal2006,dekerret2011_ISIT_journal}. Hence, $A_{i,k}^{(j)}$ can be seen as the fraction of the feedback requirements to achieve the maximal DoF.

\begin{remark} We consider here a codebook-based quantization of the channel vectors but the results can be easily translated to a setting where analog feedback is used \cite{Caire2010,ElAyach2012} by making the quantization error a function of the SNR. In fact, the digital quantization used in this work is simply a model for the errors in the channel estimates resulting from the limited feedback. Furthermore, only CSIT requirements are investigated, and different scenarios can be envisaged regarding the method used to fulfill these requirements (e.g., direct broadcasting from the RXs to all the TXs, sharing through a backhaul, \dots) \cite{Sesia2011,Rezaee2013b}. 
 \qed
\end{remark}

\section{DoF Analysis with Static Coefficients and Distributed CSI}\label{se:ST}

Let us now focus on the situation where every TX designs its precoder based on a different multi-user channel estimate. Hence, the precoding matrices used for the transmission do not form exactly an IA solution for any  imperfect estimate of the multi-user channel. This is in contrast to the centralized case studied in  \cite{Thukral2009,Krishnamachari2010}. Hence, the analysis done in these works does not hold in the setting considered here and a new approach is required.

The analysis of this situation is complicated by the fact that the function that gives the precoders as a function of the channel coefficients can not be assumed to be continuous. This can be seen by considering that there are in general multiple solutions to the IA equations \cite{Gonzalez2013}, while iterative algorithms, such as the iterative leakage minimization from \cite{Gomadam2008}, converge to one of the IA solutions. So far this convergence is not fully understood, and it can not be ruled out that a small change in the CSI (as in the case in the distributed CSI considered here) leads to a convergence to completely different solutions across the users.

Furthermore, the channel estimates at the different TXs are potentially of different accuracies such that it is not clear which accuracy dictates the DoF. Answering this question is the main goal of this work.

\subsection{Sufficient condition for an arbitrary IA scheme}
Let us denote by~$\bU_i^{\star}$ and~$\bG_i^{\star}$ the precoder and the RX filter at TX~$i$ and RX~$i$, respectively, when perfect CSIT is available at the TXs for \emph{an arbitrary IA scheme}, i.e., verifying~$(\bG_i^{\star})^{\He}\bH_{ij}\bU^{\star}_j=\bm{0}_{d_i\times d_j},\forall i \neq j$. We further define 
\begin{equation}
\bm{\Delta}\bU_{i}^{(j)}\triangleq \bU_i^{(j)}-\bU_i^{\star},\qquad \forall i,j.
\end{equation}
We now characterize the DoF achieved as a function of the precoder accuracy.
\begin{proposition}
In the IC with distributed CSIT as described in Section~\ref{se:SM}, if the CSIT is such that
\begin{equation}
\E_{\mathcal{H}}[\|\bm{\Delta}\bU_{j}^{(j)}\|_{\Fro}^2]\doteq P^{-\beta_j},\qquad \forall j, 
\end{equation}
with $\beta_j\in [0,1]$, then
\begin{equation}
\DoF_{i}\geq d_i \min_{j\neq i} \beta_j,\qquad \forall i .
\end{equation}
\label{prop_static_sufficient}
\end{proposition}
\begin{proof}
Since we want to derive a lower bound for the DoF, we can choose $\bG_k=\bG_k^{\star},\forall k$. Following a classical derivation \cite{Jindal2006,Caire2010}, we can write
\begin{equation}
R_i\geq R_i^{\star}-\E_{\mathcal{H}}\LSB\log_2\left|\I_{d_i}+P\sum_{j=1,j\neq i}^K (\bG_i^\star)^{\He}\bH_{i,j}\bU^{(j)}_{j}(\bU^{(j)}_{j})^{\He}\bH_{i,j}^{\He}\bG_i^{\star}\right |\RSB
\label{eq:Static_DoF_1}
\end{equation}
where we have defined
\begin{equation}
R_i^{\star}\triangleq \E_{\mathcal{H}}\LSB\log_2\left|\I_{d_i}+P(\bG_i^{\star} )^{\He}\bH_{i,i}\bU^{(i)}_i(\bU^{(i)}_i)^{\He}\bH_{i,i}^{\He}(\bG_i^{\star})\right|\RSB.
\label{eq:Static_DoF_2}
\end{equation}
It is easily seen that~$R_i^{\star}\doteq d_i\log_2(P)$, such that it remains to study the second term of \eqref{eq:Static_DoF_1}, which we denote by $\mathcal{I}_i$. Since~$(\bG_i^{\star})^{\He}\bH_{i,j}\bU^{\star}_{j}=\mathbf{0}_{d_i \times d_{j}}$ for $i\neq j$, it holds that
\begin{equation}
\mathcal{I}_i=\E_{\mathcal{H}}\left[\log_2\left|\I_{d_i}+P \sum_{j=1,j\neq i}^K (\bG_i^\star)^{\He}\bH_{i,j}\bm{\Delta}\bU_{j}^{(j)}(\bm{\Delta}\bU_{j}^{(j)})^{\He}\bH_{i,j}^{\He}\bG_i^{\star} \right|\right].
\label{eq:Static_DoF_3}
\end{equation}

Since~$\|\bG_i^{\star}\|^2_{\Fro}=1$, we can upper bound the interference to write
\begin{equation}
\begin{aligned}
\mathcal{I}_i&\leq\E_{\mathcal{H}}\LSB \log_2\left|\I_{d_i}+\LB P\sum_{j=1,j\neq i}^K\|\bH_{i,j}\|_{\Fro}^2\|\bm{\Delta}\bU_{j}^{(j)}\|_{\Fro}^2\RB\I_{d_i} \right |\RSB\\
&\stackrel{(a)}{\leq} d_i\LB\E_{\mathcal{H}}\LSB\log_2\left(1+P\sum_{j=1,j\neq i}^K\|\bH_{i,j}\|_{\Fro}^2\right)\RSB+\E_{\mathcal{H}}\LSB\log_2\left(1+P\sum_{j=1,i\neq j}^K \|\bm{\Delta}\bU_{j}^{(j)}\|_{\Fro}^2\right)\RSB\RB\\
&\stackrel{(b)}{\leq} d_i\LB\E_{\mathcal{H}}\LSB\log_2\left(1+P\sum_{j=1,j\neq i}^K\|\bH_{i,j}\|_{\Fro}^2\right)\RSB+\log_2\left(1+P\sum_{j=1,j\neq i}^K\E_{\mathcal{H}}\LSB\|\bm{\Delta}\bU_{j}^{(j)}\|_{\Fro}^2\RSB\right)\RB
\label{eq:Static_DoF_4}
\end{aligned}
\end{equation}
where inequality~$(a)$ can be seen to hold since only positive terms have been added and we have used Jensen's inequality to obtain inequality~$(b)$. Using that $\E_{\mathcal{H}}[\|\bm{\Delta}\bU_{j}^{(j)}\|_{\Fro}^2]\doteq P^{-\beta_{j}}$ , we can write that
\begin{equation}
\sum_{j=1,j\neq i}^K\E_{\mathcal{H}}\LSB \|\bm{\Delta}\bU_{j}^{(j)}\|_{\Fro}^2 \RSB\doteq P^{-\min_{j\neq i} \beta_{j}}.
\label{eq:Static_DoF_5}
\end{equation}
Inserting \eqref{eq:Static_DoF_5} inside \eqref{eq:Static_DoF_4} and \eqref{eq:Static_DoF_1} gives
\begin{align}
R_i&\dotgeq d_i \LB\log_2(P)-\log_2(1+P P^{-\min_{j\neq i} \beta_{j}})\RB\\
&\dotgeq d_i(\min_{j\neq i} \beta_{j}) \log_2(P),
\label{eq:Static_DoF_6}
\end{align}
which concludes the proof.
\end{proof}
Proposition~\ref{prop_static_sufficient} provides some insights into the performance by relating the accuracy with which the precoder is computed to the achieved DoF. However, the accuracy of the precoder design is difficult to relate to the accuracy of the CSIT. Indeed, the relation is dependent on the precoding method used and some precoding schemes might be more or less robust to imperfections in the CSIT. Obtaining the relation between the CSIT quality and the accuracy especially difficult to study the performance of iterative IA algorithms.

In contrast to the conventional centralized CSI configuration studied in \cite{Thukral2009,Krishnamachari2010,Rezaee2013a,Rezaee2013b}, it is not possible to study solely the IA alignment obtained at the end of the precoding scheme. Indeed, the precoders~$\bU_j^{(j)},\forall j$ do not form (a priori) together an alignment solution for any of the multi-user channel estimates available at the TXs. Hence, the structure of the IA algorithm has to be studied to observe what is the impact of the CSIT imperfection over the precoding at each TX.
\subsection{DoF analysis in the $3$-user square MIMO IC}\label{se:ST:3}
\paragraph{Perfect CSIT Solution}
We consider now a $3$-user IC with~$M_i=M,N_i=N,\forall i$ and $d_i=d,\forall i$. We also assume for the description of the IA scheme that perfect CSIT is available such that we denote the precoder used at TX~$j$ by~$\bU_j^{\star}$. Since we consider the tightly-feasible case \cite{dekerret2012_ISWCS_journal}, we have~$M=N=2d$. In that case, the IA constraints can be written as \cite{Bresler2011a}
\begin{equation}
\begin{aligned}
\mathrm{span}\left(\tilde{\mathbf{H}}_{3,1}\bU_1^{\star}\right)&=\mathrm{span}\left(\tilde{\mathbf{H}}_{3,2}\bU_2^{\star}\right),\\
\mathrm{span}\left(\tilde{\mathbf{H}}_{1,2}\bU_2^{\star}\right)&=\mathrm{span}\left(\tilde{\mathbf{H}}_{1,3}\bU_3^{\star}\right),\\
\mathrm{span}\left(\tilde{\mathbf{H}}_{2,3}\bU_3^{\star}\right)&=\mathrm{span}\left(\tilde{\mathbf{H}}_{2,1}\bU_1^{\star}\right).
\label{eq:3users_1}
\end{aligned}
\end{equation}
In particular, this system of equations can be easily seen to be fulfilled if the precoders verify
\begin{equation}
\begin{aligned}
\bU_1^{\star} \bm{\Lambda}_1&=\tilde{\mathbf{H}}_{3,1}^{-1}\tilde{\mathbf{H}}_{3,2}\tilde{\mathbf{H}}_{1,2}^{-1}\tilde{\mathbf{H}}_{1,3}\tilde{\mathbf{H}}_{2,3}^{-1}\tilde{\mathbf{H}}_{2,1}\bU_1^{\star}\\
\bU_3^{\star}&=(\tilde{\mathbf{H}}_{2,3})^{-1} \tilde{\mathbf{H}}_{2,1}\bU_1^{\star} \\
\bU_2^{\star}&=(\tilde{\mathbf{H}}_{1,2})^{-1} \tilde{\mathbf{H}}_{1,3}\bU_3^{\star}
\label{eq:3users_2}
\end{aligned}
\end{equation}
for some diagonal matrix $\bm{\Lambda}_1$. We also define for clarity the matrix~$\bY^{\star}$ equal to
\begin{equation}
\bY^{\star}\triangleq\tilde{\mathbf{H}}_{3,1}^{-1}\tilde{\mathbf{H}}_{3,2}\tilde{\mathbf{H}}_{1,2}^{-1}\tilde{\mathbf{H}}_{1,3}\tilde{\mathbf{H}}_{2,3}^{-1}\tilde{\mathbf{H}}_{2,1}.
\label{eq:3users_3}
\end{equation}
The system of equations \eqref{eq:3users_2} is then fulfilled by setting
\begin{equation}
\begin{aligned}
\bU^{\star}_1 &=\frac{1}{\sqrt{d}}\EVD(\bY^{\star})\begin{bmatrix} \bm{e}_1,\ldots, \bm{e}_d\end{bmatrix}\\
\bU^{\star}_3&=\frac{1}{\|(\tilde{\mathbf{H}}_{2,3})^{-1} \tilde{\mathbf{H}}_{2,1}\bU_1^{\star}\|_{\Fro}} (\tilde{\mathbf{H}}_{2,3})^{-1} \tilde{\mathbf{H}}_{2,1}\bU_1^{\star} \\
\bU^{\star}_2&=\frac{1}{\| (\tilde{\mathbf{H}}_{1,2})^{-1} \tilde{\mathbf{H}}_{1,3}\bU_3^{\star} \|_{\Fro}} (\tilde{\mathbf{H}}_{1,2})^{-1} \tilde{\mathbf{H}}_{1,3}\bU_3^{\star}.
\label{eq:3users_4}
\end{aligned}
\end{equation} 

\paragraph{Distributed CSIT Solution}

With distributed CSIT, TX~$j$ computes using its channel estimate~$\tilde{\bH}^{(j)}$ the matrix
\begin{equation}
\bY^{(j)}=(\tilde{\bH}_{3,1}^{(j)})^{-1}(\tilde{\bH}_{3,2}^{(j)})(\tilde{\bH}_{1,2}^{(j)})^{-1}(\tilde{\bH}_{1,3}^{(j)})(\tilde{\bH}_{2,3}^{(j)})^{-1}(\tilde{\bH}_{2,1}^{(j)}).
\label{eq:3users_4}
\end{equation}  
The precoding matrices are then obtained from
\begin{equation}
\begin{aligned}
\bU^{(j)}_1 &=\frac{1}{\sqrt{d}}\EVD(\bY^{(j)})\begin{bmatrix} \bm{e}_1,\ldots, \bm{e}_d\end{bmatrix}\\
\bU^{(j)}_3&=\frac{1}{\|(\tilde{\mathbf{H}}^{(j)}_{2,3})^{-1} \tilde{\mathbf{H}}^{(j)}_{2,1}\bU^{(j)}_1\|_{\Fro}} (\tilde{\mathbf{H}}^{(j)}_{2,3})^{-1} \tilde{\mathbf{H}}^{(j)}_{2,1}\bU_1^{(j)} \\
\bU^{(j)}_2&=\frac{1}{\| (\tilde{\mathbf{H}}^{(j)}_{1,2})^{-1} \tilde{\mathbf{H}}^{(j)}_{1,3}\bU^{(j)}_3 \|_{\Fro}} (\tilde{\mathbf{H}}^{(j)}_{1,2})^{-1} \tilde{\mathbf{H}}^{(j)}_{1,3}\bU_3^{(j)}.
\label{eq:3users_5}
\end{aligned}
\end{equation}

In that case, we can give the following result on the DoF achieved.
\begin{theorem}
Using the $3$-User IA scheme described above with distributed CSIT, the DoF achieved at user~$i$ is denoted by~$\DoF^{\DCSI}_i$ and verifies
\begin{equation}
\DoF^{\DCSI}_i \geq d\min_{j\neq i} \min_{k,\ell,k\neq \ell} A_{k,\ell}^{(j)}.
\label{eq:3users_6}
\end{equation}
\label{thm_3Users}
\end{theorem}
\begin{proof} 
The main idea of the proof is to consider only the rate achieved over the channel realizations which are ``well enough" conditioned. Over these channel realizations, the precoding is robust enough to the errors in the CSIT. Due to the continuous distribution of the channel matrices, the probability of the ``badly conditioned" channel realizations is small enough such that the loss due to removing these channel realizations can be made arbitrarily small.

We consider hereafter that $\forall k,\ell,j, A_{k,\ell}^{(j)}>0$ since the result is otherwise trivial. We also consider without loss of generality the precoding at TX~$j$. For a given~$\varepsilon>0$, we define the following channel subsets:
\begin{align}
\mathcal{X}^{\varepsilon}&\triangleq \{\tilde{\mathbf{H}}| \forall i,k,~\lambda_{\min}(\tilde{\mathbf{H}}_{i,k})\geq \varepsilon  \}\\ 
\mathcal{Y}^{\varepsilon}&\triangleq \{\tilde{\mathbf{H}}| \forall i\neq j,~|\lambda_{i}(\bY^\star)-\lambda_{j}(\bY^\star)|\geq \varepsilon  \}
\label{eq:3users_7}
\end{align} 
and 
\begin{equation}
\mathcal{H}^{\varepsilon}\triangleq \mathcal{X}^{\varepsilon} \bigcap \mathcal{Y}^{\varepsilon}.
\end{equation}
Since we aim at deriving a lower bound for the DoF (and the rate is nonnegative), we can consider only the rate achieved for the channel realizations belonging to $\mathcal{H}^{\varepsilon}$. From \eqref{eq:Static_DoF_1}, we can then write
\begin{align}
R_i&\geq \E_{\mathcal{H}^{\varepsilon}}\LSB\log_2\left|\I_{d_i}+P(\bG_i^{\star} )^{\He}\bH_{i,i}\bU^{(i)}_i(\bU^{(i)}_i)^{\He}\bH_{i,i}^{\He}(\bG_i^{\star})\right|\RSB\notag\\
&\qquad-\E_{\mathcal{H}^{\varepsilon}}\bigg[\log_2\left|\I_{d_i}+P\sum_{j=1,j\neq i}^K(\bG_i^\star)^{\He}\bH_{i,j}\bU^{(j)}_{j}(\bU_{j}^{(j)})^{\He}\bH_{i,j}^{\He}\bG_i^{\star} \right|\bigg].
\label{eq:3users_8}
\end{align} 
It can be easily seen from the continuous distribution of the channel matrices that $\forall \eta>0,\exists \varepsilon>0,\Pr(\mathcal{H}^{\varepsilon})\geq 1-\eta$. Hence, it follows that
\begin{equation}
\E_{\mathcal{H}^{\varepsilon}}\LSB \log_2\left|\I_{d_i}+P(\bG_i^{\star} )^{\He}\bH_{i,i}\bU_i^{(i)}(\bU_i^{(i)})^{\He}\bH_{i,i}^{\He}(\bG_i^{\star})\right|\RSB\dotgeq (1-\eta) d_i \log_2(P) .
\label{eq:3users_9}
\end{equation} 
We now need to upper bound the second term of \eqref{eq:3users_8} which we denote by $\mathcal{J}_i^{\varepsilon}$. We can then proceed similarly to \eqref{eq:Static_DoF_4} to write
\begin{align}
\mathcal{J}_i^{\varepsilon}&\leq  d_i\LB\E_{\mathcal{H}^{\varepsilon}}\LSB\log_2\left(1+\sum_{j=1,j\neq i}^K\|\bH_{i,j}\|_{\Fro}^2\right)\RSB+\log_2\left(1+P\sum_{j=1,j\neq i}^K\E_{\mathcal{H}^{\varepsilon}}\LSB\|\bm{\Delta}\bU_{j}^{(j)}\|_{\Fro}^2\RSB\right) \RB\\
&\dotleq d_i\LB \log_2\LB 1+P\sum_{j=1,j\neq i}^K\E_{\mathcal{H}^{\varepsilon}}\LSB \|\bm{\Delta}\bU_{j}^{(j)}\|_{\Fro}^2\RSB\RB\RB.
\label{eq:3users_10}
\end{align}  
It remains then only to compute $\E_{\mathcal{H}^{\varepsilon}}[\|\bm{\Delta}\bU_{j}^{(j)}\|_{\Fro}^2]$. Le us now consider the error due to the imperfect CSIT at TX~$j$ on one of the matrix inversion required to compute~$\bY^{(j)}$ in \eqref{eq:3users_4}. We start by introducing $\bm{\Delta}_{i,k}^{(j)}$ to represent the error done in computing the channel inverse:
\begin{equation}
\bm{\Delta}_{i,k}^{(j)}\triangleq\frac{1}{\sigma_{i,k}^{(j)}}\LB\left(\tilde{\bH}_{i,k}^{(j)}\right)^{-1}-\tilde{\mathbf{H}}_{i,k}^{-1}\RB,\qquad \forall i,k.
\label{eq:3users_11}
\end{equation}
Using the resolvent equality \cite[Lemma~$6.1$]{Couillet2011}, we can write
\begin{equation}
\bm{\Delta}_{i,k}^{(j)}=- \tilde{\mathbf{H}}_{i,k}^{-1}\mathbf{N}_{i,k}^{(j)}\tilde{\mathbf{H}}_{i,k}^{-1}+\sigma^{(j)}_{i,k} \bm{\Theta}_{i,k}^{(j)},\qquad \forall i,k,
\label{eq:3users_12}
\end{equation}
where we have defined
\begin{equation}
\bm{\Theta}_{i,k}^{(j)}\triangleq (\tilde{\mathbf{H}}_{i,k}^{(j)})^{-1}\mathbf{N}_{i,k}^{(j)}\tilde{\mathbf{H}}_{i,k}^{-1}\mathbf{N}_{i,k}^{(j)}\tilde{\mathbf{H}}_{i,k}^{-1},\qquad \forall i,k.
\label{eq:3users_13}
\end{equation}
We can then use the properties of the Frobenius norm to obtain the upper bound
\begin{align}
&\|\bm{\Delta}_{i,k}^{(j)}\|_{\Fro}\leq  \|\mathbf{N}_{i,k}^{(j)}\|_{\Fro}  \|\tilde{\mathbf{H}}_{i,k}^{-1}\|^2_{\Fro}+  \sigma_{i,k}^{(j)} \|\bm{\Theta}_{i,k}^{(j)}\|_{\Fro}.
\label{eq:3users_14}
\end{align}
Taking the expectation, we have then
\begin{equation}
\E_{\mathcal{H}^{\varepsilon}}[\|\bm{\Delta}_{i,k}^{(j)}\|^2_{\Fro}]\leq  \E_{\mathcal{H}^{\varepsilon}}\LSB \LB\|\mathbf{N}_{i,k}^{(j)}\|_{\Fro}  \|\tilde{\mathbf{H}}_{i,k}^{-1}\|^2_{\Fro}+ \sigma_{i,k}^{(j)} \|\bm{\Theta}_{i,k}^{(j)}\|_{\Fro}\RB^2\RSB
\label{eq:3users_15}
\end{equation}
The expectation in \eqref{eq:3users_15} exists and is finite because $\bH\in \mathcal{H}^{\varepsilon}$ such that the channel matrix~$\tilde{\mathbf{H}}_{i,k}$ (and~$\tilde{\mathbf{H}}^{(j)}_{i,k}$) has its eigenvalues bounded away from zero. We have therefore obtained
\begin{equation}
\E_{\mathcal{H}^{\varepsilon}} [\|\bm{\Delta}_{i,k}^{(j)}\|_{\Fro}^2]\dotleq 1.
\label{eq:3users_16}
\end{equation} 
By repeating the same calculation for every matrix inversion in \eqref{eq:3users_7}, we can write
\begin{align}
&\bY^{(j)}=(\tilde{\mathbf{H}}_{3,1}^{-1}+ \sigma_{3,1}^{(j)}\bm{\Delta}_{3,1}^{(j)})(\tilde{\mathbf{H}}_{3,2}+\sigma_{3,2}^{(j)}\bN^{(j)}_{3,2})(\tilde{\mathbf{H}}_{1,2}^{-1}+\sigma_{1,2}^{(j)}\bm{\Delta}^{(j)}_{1,2})\notag\\
&\qquad\qquad\qquad(\tilde{\mathbf{H}}_{1,3}+\sigma_{1,3}^{(j)}\mathbf{N}_{1,3}^{(j)})(\tilde{\mathbf{H}}_{2,3}^{-1}+\sigma_{2,3}^{(j)}\bm{\Delta}_{2,3}^{(j)})(\tilde{\mathbf{H}}_{2,1}+\sigma_{2,1}^{(j)}\bN^{(j)}_{2,1}).
\label{eq:3users_17}
\end{align} 
The relation \eqref{eq:3users_16} holds for every matrix inversion in \eqref{eq:3users_9} such that putting all the errors terms together, we can write from \eqref{eq:3users_17} that
\begin{equation}
\E_{\mathcal{H}^{\varepsilon}}[\|\bY^{(j)}-\bY^\star\|_{\Fro}^2]\dotleq  \max_{\ell\neq k} (\sigma_{\ell k}^{(j)})^2.
\label{eq:3users_18}
\end{equation}
Since $\bH\in \mathcal{H}^{\varepsilon}$, all the eigenvalues of~$\bY^\star$ (and $\bY^{(j)}$) are different and the matrices~$\bY^{\star}$ and $\bY^{(j)}$ are diagonalizable. Let $\bY^\star=\bV^\star\bm{\Lambda}(\bV^\star)^{\He}$, with $\bV^\star \in \mathbb{C}^{M \times M}$ and~$\bm{\Lambda}^\star\in \mathbb{C}^{M\times M}$, and $\bY^{(j)}=\bV^{(j)}\bm{\Lambda}^{(j)}(\bV^{(j)})^{\He}$, with $\bV^{(j)} \in \mathbb{C}^{M \times M}$ and~$\bm{\Lambda}^{(j)}\in \mathbb{C}^{M\times M}$, be the spectral decomposition of~$\bY^\star$ and~$\bY^{(j)}$, respectively. Applying Theorem~$2.1$ from \cite{Chen2012} to $\bY^\star$ and $\bY^{(j)}$ and taking the expectation we can show that for some constant~$\gamma^{(j)}>0$ independant of the SNR~$P$,
\begin{align}
\E_{\mathcal{H}^{\varepsilon}}[\|\bV^{(j)}-\bV^\star\|_{\Fro}^2]&\leq \gamma^{(j)} \E_{\mathcal{H}^{\varepsilon}}[\|\bY^{(j)}-\bY^\star\|_{\Fro}^2]\\
&\dotleq \max_{\ell\neq k} (\sigma_{\ell ,k}^{(j)})^2\label{eq:3users_7_1}\\
&\doteq P^ {-\min_{\ell\neq k} A_{\ell, k}^{(j)}}.
\label{eq:3users_19}
\end{align}
Let us denote by $\bar{\bV}^{(j)}$ and $\bar{\bV}^\star$ the matrices made of the first $d$ columns of $\bV^{(j)}$ and $\bV^\star$, respectively. 
The precoding scheme is such that~$\bU_1^{\star}=\bar{\bV}^{\star}$ and ~$\bU_1^{(1)}=\bar{\bV}^{(1)}$. Hence,
\begin{align}
\E_{\mathcal{H}^{\varepsilon}}[\| \bm{\Delta}\bU_1^{(1)}\|_{\Fro}^2] \dotleq  P^ {-\min_{\ell\neq k} A_{\ell,k}^{(1)}}.
\label{eq:3users_20}
\end{align} 
The relation~\eqref{eq:3users_20} is easily extended to the other precoders~$\bU^{(2)}_2$ and $\bU^{(3)}_3$ to obtain that
\begin{align}
\sum_{j=1,j\neq i}^3\E_{\mathcal{H}^{\varepsilon}}[\|\bm{\Delta}\bU_{j}^{(j)}\|_{\Fro}^2] &\dotleq \sum_{j=1,j\neq i}^3 P^ {-\min_{\ell\neq k} A_{\ell,k}^{(j)}}\\
&\dotleq P^ {-\min_{j\neq i}\min_{\ell\neq k} A_{\ell,k}^{(j)}}.
\label{eq:3users_21}
\end{align}
Coming back to \eqref{eq:3users_8}, this gives
\begin{align}
R_i&\dotgeq  d_i\LB (1-\eta)\log_2(P)-\log_2(1+P P^ {-\min_{j\neq i}\min_{\ell\neq k} A_{\ell,k}^{(j)}} \RB\\
&\dotgeq  d_i\LB\min_{j\neq i}\min_{\ell\neq k} A_{\ell,k}^{(j)} - \eta\RB \log_2(P).
\label{eq:3users_22}
\end{align} 
Choosing $\eta$ arbitrarily small concludes the proof.
\end{proof}

We have shown that for the $3$-user IA closed-form alignment scheme, the achieved DoF is larger than the worst accuracy of the channel estimates across the TXs. Note that this lower bound is in fact conjectured to be tight. 

Interestingly, the lower bound at RX~$j$ is limited by the accuracy of the estimates relative to the channels of all the \emph{other} RXs. This result is in strong contrast with the centralized setting where the DoF of user~$i$ depends \emph{solely} on the accuracy with which the channel matrices from the TXs to RX~$i$ are fedback. This show how IA becomes more sensitive to CSIT errors when the precoding is done based on distributed CSIT. Note that this result is reminiscent of \cite{dekerret2011_ISIT_journal} where it was shown in a $K$-user MISO BC with single-antenna RXs and with distributed CSIT, that the DoF was limited by the worst accuracy across the TXs and across the channel vectors.  

\section{Simulations}\label{se:Simulations}
In this section, we validate by Monte-Carlo simulations the results in the $3$-user square IC channel studied in Subsection~\ref{se:ST:3}. We consider $M= N=4$ and $d=2$ and we average the performance over $10000$ realizations of a Rayleigh fading channel. We consider the distributed CSIT configuration described in Section~\ref{se:SM}. The quantization error is modeled using \eqref{eq:SM_7} with $(\sigma_{i,k}^{(j)})^2= 2^{-B_{i,k}^{(j)}/(N_iM_j-1)}$ and $\bN_{i,k}^{(j)}$ having its elements i.i.d.~$\CN(0,1)$. We choose the CSIT scaling coefficients as
\begin{equation}
\begin{aligned}
&\forall (i,k,j)\in \{1,2,3\}^{3}\setminus \{(3,2,2),(3,2,3)\},~A_{i,k}^{(j)}=1, A_{3,2}^{(2)}=0.5,~A_{3,2}^{(3)}=0.
\label{eq:sim_1}
\end{aligned}
\end{equation}
Following Theorem~\ref{thm_3Users}, we have for the CSIT configuration described in \eqref{eq:sim_1} that $\DoF_1\geq 0$, $\DoF_2\geq 0$, and~ $\DoF_3\geq 0.5 d=1$. The average rate achieved is shown for each user in Fig.~\ref{R_3Users_Ex1}. For comparison, we have also simulated the average rate per-user achieved based on perfect CSIT and with distributed CSIT when the CSIT scaling coefficients are set equal to $1$ for every TX ($\forall i,k,j, A_{i,k}^{(j)}=1$). It can then be verified that having all CSIT scaling coefficients equal to one allows to achieve the maximal DoF. 

With the CSIT configuration described in \eqref{eq:sim_1}, the slope of the rate of user~$3$ decreases as the SNR increases, revealing a very slow convergence to the DoF. This makes it difficult to accurately observe the DoF achieved. Yet, it can be seen that having only~$A_{3,2}^{(3)}$ equal to zero leads already to the saturation of the rates of users~$1$ and $2$ (i.e., their DoF is equal to $0$), which tends to confirm our conjecture.
\begin{figure}
\centering
\includegraphics[width=1\columnwidth]{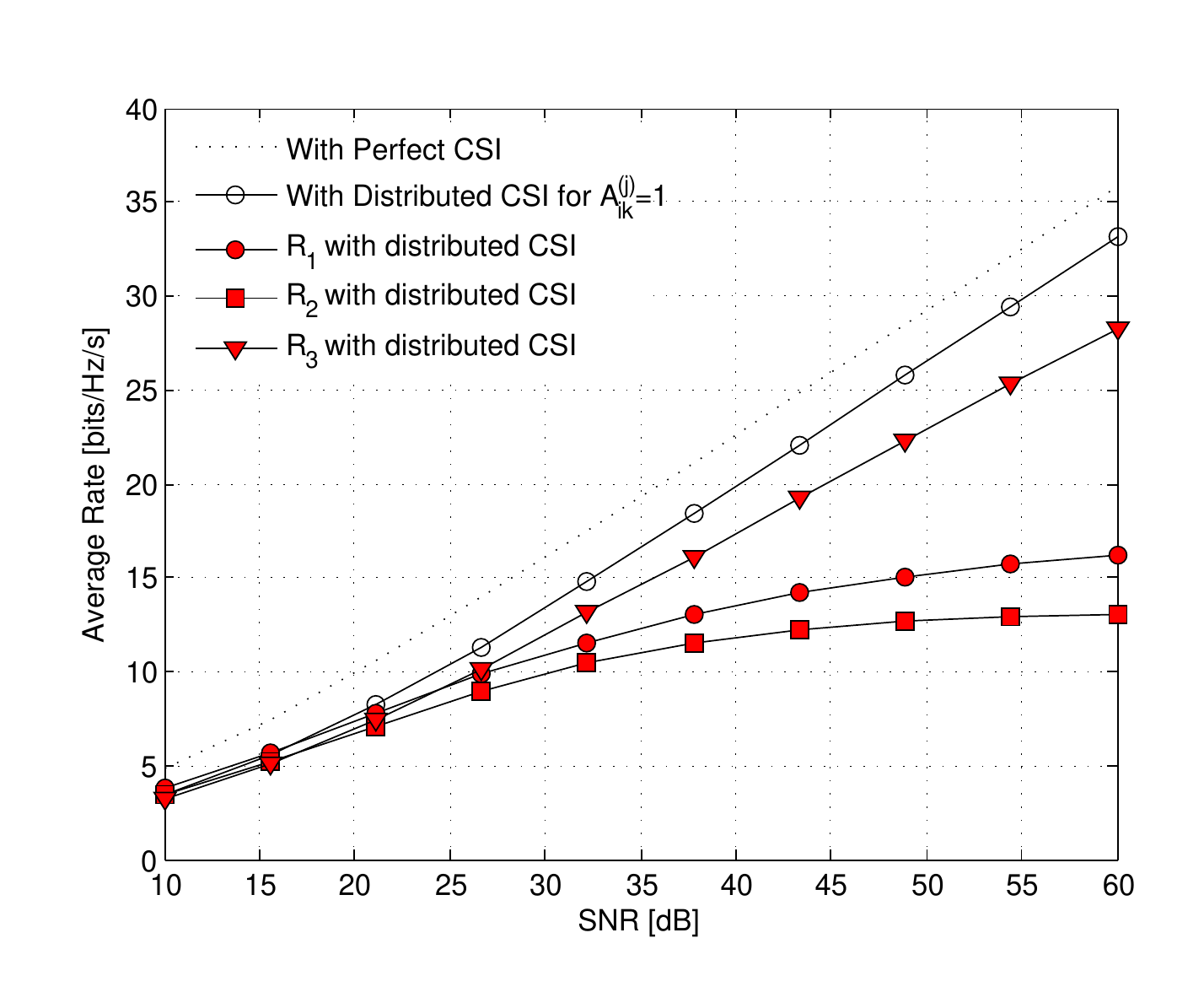}
\caption{Average rate per user in the square setting $M=N=4$ with $d=2$ for the CSIT scaling coefficients given in \eqref{eq:sim_1}.}
\label{R_3Users_Ex1}
\end{figure}

\section{Extension to Time-Alignment and Iterative Interference Alignement}
We have studied the DoF in a particular antenna configuration for the case of static MIMO channels. This antenna configuration has been considered both because it is believed to be a simple, yet practically relevant configuration, and because the knowledge of a closed-form precoding formula is necessary for our analysis. In fact, our approach is expected to easily extend to numerous scenarios where a closed form expression exists for the IA precoding, under the condition that the precoding scheme is ``robust" enough to the quantization errors, e.g., it consists of matrix inversions or matrix multiplications where the matrices have their elements distributed according to a continuous distribution. This in particular the case of the original time-alignment IA scheme from \cite{Cadambe2008,Gou2010}. Hence, our results can be trivially extended to this setting.

Obtaining the DoF achieved with an iterative IA algorithm like the min-leakage algorithm or the max-SINR algorithm \cite{Gomadam2008,Peters2011} is a challenging open problem which will be investigated in subsequent works. As a prerequisite step, it requires deriving some basic properties of the IA algorithm, such as convergence properties, which have remained out of reach until now. Furthermore, it has been shown in \cite{dekerret2012_ISWCS_journal} considering the different model of \emph{incomplete CSIT} that heterogeneous antenna configurations could be exploited to achieve IA even when some of the TXs do not have any CSIT. Such behaviour should be taken into account when analysing the feedback requirements and complicate further the analysis of iterative IA algorithms.
\bibliographystyle{IEEEtran}
\bibliography{Literatur}
\end{document}